\documentclass[journal,onecolumn]{IEEEtran}
\usepackage{latexsym,epsfig,subfigure,amssymb,dsfont}
\usepackage{amsmath}
\newtheorem{definition}{Definition}[section]
\newtheorem{remark}{Remark}[section]
\newtheorem{theorem}{Theorem}[section]

\newcommand{\eqdef} {\mbox{$\:\stackrel{\triangle}{=}\:$}}

\newcommand{\bq}{{\mathbf{q}}}
\newcommand{\bx}{{\mathbf{x}}}

\newcommand{\by}{{\mathbf{y}}}

\newcommand{\bZ}{{\mathbf{Z}}}
\newcommand{\bu}{{\mathbf{u}}}
\newcommand{\bv}{{\mathbf{v}}}

\newcommand{\bs}{{\mathbf{s}}}
\newcommand{\bY}{{\mathbf{Y}}}

\newcommand{\bbZ}{{\mathbb{Z}}}
\newcommand{\bbR}{{\mathbb{R}}}
\newcommand{\bS}{{\mathbf{S}}}

\newcommand{\cM}{{\mathcal{M}}}

\newcommand{\cN}{{\mathcal{N}}}

\newcommand{\cC}{{\mathcal{C}}}
\newcommand{\cS}{{\mathcal{S}}}
\newcommand{\cQ}{{\mathcal{Q}}}
\newcommand{\cR}{{\mathcal{R}}}

\newcommand{\cU}{{\mathcal{U}}}
\newcommand{\cV}{{\mathcal{V}}}
\newcommand{\cW}{{\mathcal{W}}}
\newcommand{\hW}{{\hat{W}}}
\newcommand{\hM}{{\hat{M}}}
\newcommand{\hq}{{\hat{q}}}
\newcommand{\hQ}{{\hat{Q}}}

\newcommand{\hbq}{{\hat{\mathbf{q}}}}

\newcommand{\cX}{{\mathcal{X}}}
\newcommand{\cY}{{\mathcal{Y}}}

\newcommand{\ep}{{\epsilon}}

\newcommand{\be}{\begin{equation}}
\newcommand{\ee}{\end{equation}}
\newcommand{\bea}{\begin{eqnarray}}
\newcommand{\eea}{\end{eqnarray}}
\newcommand{\bean}{\begin{eqnarray*}}
\newcommand{\eean}{\end{eqnarray*}}
\newcommand{\ben}{\begin{enumerate}}
\newcommand{\een}{\end{enumerate}}

    \def\squarebox#1{\hbox to #1{\hfill\vbox to #1{\vfill}}}

\begin{document}

\bibliographystyle{ieeetr}
\baselineskip=2.0\normalbaselineskip

\title{An Information Theoretic Analysis of Single Transceiver Passive RFID Networks}

\author{Y\"{u}cel Altu\u{g},~\IEEEmembership{}
        S. Serdar Kozat,~\IEEEmembership{Member}
        M. K{\i}van\c{c} M{\i}h\c{c}ak~\IEEEmembership{Member}%
\thanks{Y.~Altu\u{g} is with the School of Electrical and Computer Engineering of Cornell University, Ithaca, NY, 14853, USA. (e-mail: ya68@cornell.edu),  
M.~K.~M{\i}h\c{c}ak are with the Electrical and Electronic Engineering Department of Bo\u{g}azi\c{c}i University,
Istanbul, 34342, Turkey (e-mail: kivanc.mihcak@boun.edu.tr), S.~S.~Kozat is with the Electrical and
Electronic Engineering Department of Ko\c{c} University, Rumeli
Feneri Yolu, Sar{\i}yer, Istanbul, 34450, Turkey (e-mail: skozat@ku.edu.tr)}%
\thanks{Y.~Altu\u{g} is partially supported by T\"{U}B\.{I}TAK Career
Award no. 106E117; M.~K.~M{\i}h\c{c}ak is partially supported by
T\"{U}B\.{I}TAK Career Award no. 106E117 and T\"{U}BA-GEBIP Award.}}
\maketitle

\begin{abstract}
In this paper, we study single transceiver passive RFID networks by modeling the
underlying physical system as a special cascade of a certain broadcast
channel (BCC) and a multiple access channel (MAC), using a ``nested
codebook'' structure in between. The particular application
differentiates this communication setup from an ordinary cascade of a
BCC and a MAC, and requires certain structures such as ``nested
codebooks'', impurity channels or additional power
constraints. We investigate this problem both for discrete alphabets,
where we characterize the achievable rate region, as well as for
continuous alphabets with additive Gaussian noise, where we provide the capacity region. 
Hence, we establish the maximal achievable error free communication rates for
this particular problem which constitutes the fundamental limit that is
achievable by any TDMA based RFID protocol and the achievable rate
region for any RFID protocol for the case of continuous alphabets under additive Gaussian 
noise.
\end{abstract}

\section{Introduction}
\label{sec:intro}


 In this paper, we deal with a multiuser communication setup which
consists of ``cascade'' of a broadcast channel (BCC) and a multiple
access channel (MAC). The encoder of BCC part and the decoder of the
MAC part is the same transceiver, and the decoders of the BCC part
and the encoders of the MAC part are the mobile units of the system.
The ultimate goal of the communication system considered in the
paper is the following: transceiver\footnote{In practical RFID systems, the problem of reader collusion is also considered, which amounts to having multiple transceivers in our setup. In our case, we concentrate on the ``single reader (transceiver)'' setup as first step.} wants to ``find out'' some
specific information possessed by the mobile units and for this
purpose it first broadcasts the ``type'' of the information it seeks
to receive from each mobile unit. Then every mobile unit
``sends'' the corresponding information of the received type to the
transceiver. The specific type of information phenomenon
differentiates the system at hand from the ordinary cascade of BCC
and MAC, because in order to model this situation we employ a
\emph{nested codebook structure} at the MAC encoders, i.e. at the
mobile units, which will be explained in detail in
Section~\ref{ssec:problem-statement}.


Beyond its promising structure to model wireless communication
networks, the problem at hand gives the fundamental limits of RFID
protocols in two different ways, supposing the transceiver is RFID
reader, mobile units are RFID tags and the RFID reader knows the set
of the IDs of the RFID tags in the environment:
\begin{enumerate}
\item[(i)] The above mentioned communication problem gives the
fundamental limits achievable in TDMA based RFID protocols, since
the transceiver sends the TDMA time slots, which are designated to
allow communication in a collusion free manner, using the BCC part
and then mobile units uses their corresponding time slot information
in order to transmit their data to the RFID reader. Supposing equal
information rate, say $R^{ID}$, at each BCC branch, the maximum
number of RFID tags that can be handled is $2^{R^{ID}}$ and the
maximum data rate from tags to reader is the maximum rate that can
be achieved using TDMA at the MAC part of the communication system.

\item[(ii)] The above mentioned communication problem gives the
fundamental limits of any RFID protocol, since the RFID reader
transmits ``on-off'' message\footnote{This on-off message also
meaningful in practice as far as passive RFID tags are concerned,
since they need to facilitate an external energy in order to
operate} from the BCC to tags, and then tags communicate back their
data through the MAC simultaneously to the reader. The achievable
rate region of the MAC part is the fundamental limit of any RFID
protocol under the assumption that receiver knows the set of the IDs
of the RFID tags in the environment.
\end{enumerate}


The nested codebook structure used in the MAC part of this paper is
similar to the ``pseudo users'' concept introduced in \cite{Med:04},
where the authors investigate a special notion of capacity for time
slotted ALOHA systems by combining multiple access rate splitting and
broadcast codes.  However, in \cite{Med:04}, the authors explicitly
investigate the ALOHA protocol over a degraded additive Gaussian noise
channel, where users communicate over a common channel using data
packets with predefined collusion probability.  Unlike
\cite{Med:04}, our codes achieve the capacity in the usual sense,
where the codewords are sent with arbitrarily small error probability. We
also investigate a cascade structure including a BCC in the front and
a different MAC in the end. We study this setup both for discrete
alphabets using imperfection channels to model the impurities of the
actual physical system as well as for continuous alphabets over additive 
Gaussian noise channel by including appropriate power constraints.

We note that the nested codebook structure used in this paper differs
from the nested codes defined in \cite{Zam:02,Bar:03}. In
\cite{Zam:02} nested codebooks, especially nested lattices codes, are
explicitly defined with a multi-resolution point of view, where the
nesting of codes provide progressively coarser description to finer
description of the intended information. Here, our nested codebooks
are independent from each other and convey different information.

Organization of the paper is as follows: In
Section~\ref{sec:notation-problem-statement} we state the notation
followed throughout the paper and formulate the communication
problem considered in the paper. Section~\ref{sec:discrete-case}
devoted to derive an achievable rate region of the problem for the
case of discrete alphabets, by also including ``imperfection
channels'' in order to model the practical phenomenon better. In
Section~\ref{sec:gaussian-case}, we state the capacity region of the
problem for the case of Gaussian BCC and Gaussian MAC by also
incorporating suitable power constraints. Paper ends with the
conclusions given in Section~\ref{sec:conclusion}.

\section{Notation and Problem Statement}
\label{sec:notation-problem-statement}

\subsection{Notation}
\label{ssec:notation}

Boldface letters denote vectors; regular letters with subscripts
denote individual elements of vectors. Furthermore, capital letters
represent random variables and lowercase letters denote individual
realizations of the corresponding random variable. The sequence of
$\left\{ a_1, a_2, \ldots , a_N \right\}$ is compactly represented
by $\mathbf{a}^N$. The abbreviations ``i.i.d.'', ``p.m.f.'' and
``w.l.o.g.'' are shorthands for the terms ``independent identically
distributed'', ``probability mass function'' and ``without loss of
generality'', respectively.

\subsection{Problem Statement}
\label{ssec:problem-statement}

In this paper, our major concern is finding maximum achievable
error-free rates for the following multiuser communication problem
(For the sake of simplicity, we define the problem for the case of
two mobile units, however all of the results can easily be
generalized to $M$ users using the same arguments employed in the
paper): A transceiver first acts as a transmitter and
\emph{broadcasts} a pair of messages, $(W_1,W_2) \in \cW_1 \times
\cW_2$, to mobile units through the first memoryless communication
channel. Mobile units decode the messages intended to them, i.e.
first (resp. second) mobile unit decides $\hW_1$ (resp. $\hW_2$),
and then choose their messages accordingly, i.e. first (resp.
second) mobile unit chooses $M_1 \in \cM_1^{\hW_1}$ (resp. $M_2 \in
\cM_2^{\hW_2}$), and \emph{simultaneously} sends to transceiver,
which this time acts as a receiver, through the second memoryless
communication channel.

Next, we give the quantitative definition of the communication
system considered:

\begin{definition}
The above-mentioned communication system consists of the following
components:
\begin{enumerate}
\item[(i)] Eight discrete finite sets $\cX$, $\cY_1$, $\cY_2$,
$\cQ_1$, $\cQ_2$, $\hat{\cQ_1}$, $\hat{\cQ_2}$, $\cS$.

\item[(ii)] A one-input two-output, discrete memoryless communication
channel, termed as ``broadcast channel part'' or shortly BCC part
from now on, modeled by a conditional p.m.f. $p(y_1,y_2|x) \in \cY_1
\times \cY_2 \times \cX$. Using the memoryless property, we have the
following expression for the n-th extension of the BCC part:
\begin{equation}
p(\by_1^n, \by_2^n|\bx^n) = \prod_{k=1}^n p(y_{1k},y_{2k}|x_k).
\label{eq:problem-statement-1}
\end{equation}

\item[(iii)] The memoryless ``imperfections channel'', which models
the impurities and the instantaneous erroneous behavior at the
mobile units (especially useful in the modeling of the RFID tags),
given by a conditional p.m.f. $p(\hq_i|q_{i}) \in \hat{\cQ} \times
\cQ_i$. Using the memoryless property, we have the following
expression for the n-th extension of the i-th imperfection channel
\begin{equation}
p( \hbq^n_i|\bq^n_{i} ) = \prod_{k=1}^n p(\hq_{i,k} | q_{i,k}),
\label{eq:problem-statement-2}
\end{equation}
for $i \in \{1,2\}$.

\item[(iv)] A two-input one-output, discrete memoryless communication
channel, termed as ``multiple access channel part'' or shortly MAC
part from now on, given by a conditional p.m.f. $p(s|\hq_1,\hq_2)
\in \cS \times \hat{\cQ_1} \times \hat{\cQ_2}$. Using the memoryless
property, we have the following expression for the n-th extension of
the MAC part:
\begin{equation}
p(\bs^n | \hbq_1^n, \hbq_2^n) = \prod_{k=1}^n
p(s_k|\hq_{1,k},\hq_{2,k}). \label{eq:problem-statement-3}
\end{equation}
\end{enumerate}
\label{def:communication-system}
\end{definition}

Next, we state the code definition
\begin{definition}
An $\left( 2^{nR_1^{ID}}, 2^{nR_2^{ID}}, 2^{nR_1^{Data}},
2^{nR_2^{Data}}, n \right)$ code for the communication system given
above consists of the following parts:

\begin{enumerate}
\item[(i)]Pair of transmitter messages, termed as ``broadcast channel messages''
or shortly BCC messages from now on, to mobile units given as
$(W_1,W_2) \in \cW_1 \times \cW_2$, where $\cW_i \eqdef \left\{ 1,
\ldots, 2^{nR_i^{ID}}\right\}$ for $i \in \{1,2\}$.

\item[(ii)] The transceiver's encoding function, termed as ``broadcast channel encoder'' or shortly BCC encoder from now
on, given as
\begin{equation}
X^{BCC} \; : \; \cW_1 \times \cW_2 \rightarrow \cX^n, \textrm{ such
that } X^{BCC}\left(W_1,W_2\right) = \bx^n(W_1,W_2).
\end{equation}

\item[(iii)] The mobile units' decoding functions, termed as ``broadcast channel decoders'' or shortly BCC decoders from now
on, given by $g_i^{BCC} \; : \; \cY_i^n \rightarrow \cW_i \cup
\{0\}$, such that $g_i^{BCC}(\bY_1^n) = \hW_i$, for $i \in \{1,2\}$,
where $\{0\}$ corresponds to ``miss-type'' error event.

\item[(iv)] The mobile units' messages corresponding to decoded BCC messages $\hW_i$, termed as
``multiple access channel messages'' or shortly MAC messages from
now on, $M_i \in \cM_i^{\hW_i}$, where $\cM_i^{\hW_i} \eqdef \left\{
1, \ldots, 2^{nR_i^{Data}}\right\}$, for $i \in \{1,2\}$. Note that
this is the message part of a ``nested codebook structure''
corresponding to the decoded message $\hW_i$ at each mobile unit.

\item[(v)] The mobile units' encoding function, termed as ``multiple access channel encoders''
or shortly MAC encoders from now on, given by $Q_{i}^{MAC} \; : \;
\cM_i^{\hW_i} \rightarrow \cQ_i^n$, for $i \in \{ 1,2\}$, such that
$Q_i^{MAC}(M_i) = \bq_{\hW_i}^n\left(M_i\right)$. Note that
$\bq_{\hW_i}^n\left(M_i\right)$'s are the codewords of the ``nested
codebook structure'' corresponding to the decoded message $\hW_i$ at
each mobile unit.

\item[(vi)] The transceiver's decoding function, termed as ``multiple access channel decoder'' or shortly MAC decoder from now
on, given by $g^{MAC} \; : \; \cS^n \rightarrow \cM_1^{\cW_1} \times
\cM_2^{\cW_2}$.

\item[(vii)] Decoded messages at the transceiver: $\left(\hM_1, \hM_2 \right) \in \cM_1^{W_1} \times
\cM_2^{W_2}$. Note that since transceiver knows $(W_1,W_2)$ pair and
tries to ``learn'' the corresponding $(M_1,M_2)$ pairs
simultaneously, hence it chooses $(M_1,M_2)$-th messages from the
set $\cM_1^{W_1} \times \cM_2^{W_2}$.
\end{enumerate}
\label{def:code}
\end{definition}

Obviously, the communication system may be intuitively considered as
a cascade of a two user ``broadcast channel''\cite{cov:06} and a two
user ``multiple access channel''\cite{cov:06} with the following
modifications: first the employment of the nested codebook structure
at the MAC encoders and the imperfections channels included. The
aforementioned modified cascade, including the encoders, codewords
and decoders at both BCC and MAC part is shown in
Figure~\ref{fig:discrete} below:

\begin{figure}[!htb]
\centering \epsfxsize 7in
 \epsfbox{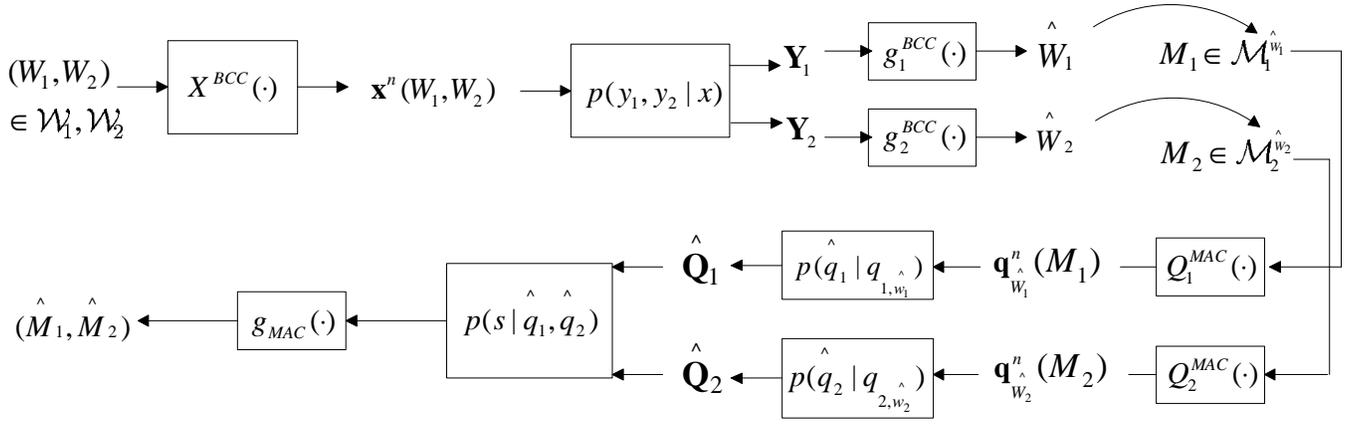} \caption{Block Diagram Representation of the multiuser communication system considered in the paper.} \label{fig:discrete}
\end{figure}

Now, we state following ``probability of error'' related
definitions, which will be used throughout the paper.

\begin{definition}
$ $\\
\vspace{-1cm}
\begin{itemize}
\item[(i)] The \emph{conditional probability of error}, $\lambda_i$,
for the communication system is defined by:
\begin{equation}
\lambda_{w_1,w_2,m_1,m_2}  \eqdef  1- \Pr\left( \left[ (\hW_1,\hW_2)
= (w_1,w_2) | (W_1,W_2) = (w_1,w_2) \right] \wedge \left[
(\hM_1,\hM_2) = (m_1,m_2) | (M_1,M_2) = (m_1,m_2) \right] \right),
\label{eq:problem-statement-4-cond-prob-err}
\end{equation}
and the \emph{maximal probability of error}, $\lambda^{(n)}$, for
the communication system is defined by:
\begin{equation}
\lambda^{(n)}  \eqdef  \max_{w_1,w_2,m_1,m_2}
\lambda_{w_1,w_2,m_1,m_2}.
\label{eq:problem-statement-5-max-prob-err}
\end{equation}

\item[(ii)] The \emph{conditional probability of error for the BCC part},
$\lambda_i^{BCC}$, is defined by:
\begin{equation}
\lambda_{BCC}^{w_1,w_2} \eqdef \Pr\left( (\hW_1,\hW_2) \neq
(w_1,w_2) | (W_1,W_2) = (w_1,w_2) \right),
\label{eq:cond-prob-err-BCC}
\end{equation}
and the \emph{average probability of error for the BCC part},
$P_{e,BCC}^{(n)}$, is defined by:
\begin{equation}
P_{e,BCC}^{(n)} \eqdef \Pr\left( \left( \hW_1, \hW_2 \right) \neq
\left( W_1, W_2\right)\right), \label{eq:average-prob-err-BCC}
\end{equation}

\item[(iii)] The \emph{conditional probability of error for the MAC part},
$\lambda_i^{MAC}$, is defined by:
\begin{equation}
\lambda_{MAC}^{m_1,m_2} \eqdef \Pr\left( (\hM_1,\hM_2) \neq
(m_1,m_2) | (M_1,M_2) = (m_1,m_2), (\hW_1,\hW_2) = (w_1,w_2)\right),
\label{eq:cond-prob-err-MAC}
\end{equation}
and the \emph{average probability of error for the MAC part},
$P_{e,MAC}^{(n)}$ is defined by:
\begin{equation}
P_{e,MAC}^{(n)} \eqdef \Pr\left( \left( \hM_1, \hM_2 \right) \neq
\left( M_1, M_2 \right) | \left( \hW_1, \hW_2 \right) = \left( w_1,
w_2 \right) \right). \label{eq:average-prob-err-MAC}
\end{equation}
\end{itemize}
\label{def:err-prob}
\end{definition}

Note that, using
\eqref{eq:problem-statement-4-cond-prob-err},\eqref{eq:cond-prob-err-BCC}
and \eqref{eq:cond-prob-err-MAC} we conclude that
\begin{equation}
\lambda_{w_1,w_2,m_1,m_2} = 1 - (1 - \lambda_{BCC}^{w_1,w_2})(1 -
\lambda_{MAC}^{m_1,m_2}). \label{eq:cond-prob-err-overall}
\end{equation}

Next, \emph{achievability} is defined as
\begin{definition}
Any rate quadruple $\left( R_1^{ID}, R_2^{ID}, R_1^{Data},
R_2^{Data} \right)$ is said to be \emph{achievable} if there exists
a sequence of codes $\left(2^{nR_1^{ID}}, 2^{nR_2^{ID}},
2^{nR_1^{Data}}, 2^{nR_2^{Data}}, n\right)$ such that $\lambda^{(n)}
\rightarrow 0$ as $n \rightarrow \infty$.
\label{def:achievable-rates}
\end{definition}

\section{Discrete Case}
\label{sec:discrete-case} In this section, we deal with the problem
stated in Section~\ref{ssec:problem-statement} under the discrete
random variables assumption.

\subsection{Achievable Region for The General Case}
\label{ssec:achievable-region-discrete-case} The main result of this
section is the following theorem:
\begin{theorem}
\emph{(Achievability-Discrete Case)} Any quadruple $\left( R_1^{ID},
R_2^{ID}, R_1^{Data}, R_2^{Data}\right) \in \cR_0$ is achievable,
where
\begin{eqnarray}
\cR_0 \eqdef \left\{ ( R_1^{ID}, R_2^{ID}, R_1^{Data}, R_2^{Data} ) \; : \; R_1^{ID}, R_2^{ID}, R_1^{Data}, R_2^{Data} \geq 0, \; R_1^{ID} < I \left( U; Y_1 \right), \; R_2^{ID} < I \left( V; Y_2 \right), \right. \nonumber \\
R_1^{ID}+R_2^{ID} < I \left( U; Y_1 \right) + I \left( V; Y_2
\right) - I \left( U; V \right), \; R_1^{Data} < I ( \hQ_1 ; S |
\hQ_2 ), \; R_2^{Data} < I ( \hQ_2 ; S | \hQ_1 ), \nonumber \\
R_1^{Data}+R_2^{Data} < I (\hQ_1,\hQ_2 ; S), \textrm{ for some }
p(u,v,x) \textrm{ on } \cU
\times \cV \times \cX \textrm{ and } p(q_1,q_2,s) \textrm{ on } \cQ_1 \times \cQ_2 \times \cS,\nonumber\\
\left. \textrm{ where } p(q_1,q_2,s) \eqdef \sum_{\hq_1,\hq_2}
p(s|\hq_1,\hq_2) p(\hq_1|q_1) p(\hq_2|q_2) p(q_1) p(q_2), \textrm{
for some } p(q_1),p(q_2) \textrm{ on } \cQ_1, \cQ_2, \textrm{
respectively}\right\} .\label{eq:discrete-case-1-ach-region}
\end{eqnarray}
\label{thrm:discrete-case-achievable-region}
\end{theorem}

\begin{proof}
Proof follows combining arguments from \cite{gam-meu:81} and
\cite{cov:06} for BCC and MAC parts, respectively; by also taking
imperfection channels and nested codebook structure into account.

W.l.o.g. we suppose $\epsilon \in (0,1)$. \footnote{Since we want to show that $\lambda^{(n)} \rightarrow 0$ as $n \rightarrow \infty$, this will suffice. To see this, observe that in the proof of the theorem, we show that for any sufficiently large $n$ and for any $\epsilon \in (0,1)$, $\lambda^{(n)} \leq \epsilon$, which directly implies  $\lambda^{(n)} \leq \epsilon^\prime$ for any $\epsilon^{\prime} \geq 1$.}

First, define $A_{\ep}^{(n)}(U)$ (resp. $A_{\ep}^{(n)}(V)$) as the
set of $\ep$-typical sequences \cite{cov:06} $\mathbf{u}^n \in
\cU^n$ (resp. $\mathbf{v}^n \in \cV^n$) for any given $p(u)$ (resp.
$p(v)$) on $\cU$ (resp. $\cV$).

Next, for $w_1 \in \{1, \ldots, 2^{nR_1^{ID}} \}$, we define
following cells:
\[
B_{w_1} \eqdef \left[ (w_1 - 1) 2^{n(I(U;Y_1)-R_1^{ID}-\ep)} + 1 ,
\; w_1 2^{n(I(U;Y_1)-R_1^{ID}-\ep)}\right].
\]
Similarly, for resp. $w_2 \in \{1, \ldots, 2^{nR_2^{ID}} \}$, we
define:
\[
 C_{w_2} \eqdef \left[ (w_2 - 1) 2^{n(I(V;Y_2)-R_2^{ID}-\ep)} + 1 ,
\; w_2 2^{n(I(V;Y_2)-R_2^{ID}-\ep)}\right],
\]
w.l.o.g. supposing that
$2^{n(I(U;Y_1)-R_1^{ID}-\ep)},2^{n(I(V;Y_2)-R_2^{ID}-\ep)} \in
\bbZ^+$.

\textbf{Encoding at BCC part:}
\begin{itemize}

\item[i)]\underline{Generation of the codebook}:
Generate the codebook $\mathcal{C}_{BCC} \in \cX^{2^{n R^{ID}_1}}
\times \cX^{2^{n R^{ID}_2}} \times \cX^n$ such that $(i,j,m)$-th
element is $x_m(i,j)$ and $x_m(i,j)$s are i.i.d. realizations of $X$
of which distribution is $p(x) = \sum_{u,v} p(u,v,x)$ for all
$i,j,m$ and reveal the codebook to both mobile units and
transceiver.

\item[ii)] Choose an $(W_1,W_2) \in \cW_1 \times \cW_2$ uniformly over $\cW_1
\times \cW_2$, i.e. $\Pr(W_1 = w_1, W_2 = w_2) = 1/\left(
2^{nR_1^{ID}}2^{nR_2^{ID}} \right)$, for all $(w_1,w_2) \in \cW_1
\times \cW_2$.

\item[iii)] Next, generate $2^{n(I(U;Y_1) - \ep)}$, i.i.d. $\bu^n$, such that
\[
p(\bu^n)=\left\{
           \begin{array}{cl}
             \frac{1}{|A_\ep^{(n)}(U)|} & \hbox{, if } \bu^n \in A_\ep^{(n)}(U) \\
             0 & \hbox{, otherwise}
           \end{array}
         \right.
\]
Similarly, generate $2^{n(I(V;Y_2) - \ep)}$, i.i.d. $\bv^n$, such
that
\[
p(\bv^n)=\left\{
           \begin{array}{cl}
             \frac{1}{|A_\ep^{(n)}(V)|} & \hbox{, if } \bv^n \in A_\ep^{(n)}(V) \\
             0 & \hbox{, otherwise}
           \end{array}
         \right.
\]
Label these $\bu^n(k)$ (resp. $\bv^n(l)$), $k \in \left[ 1,
2^{n(I(U;Y_1) - \ep)}\right]$ (resp. $l \in \left[ 1, 2^{n(I(V;Y_2)
- \ep)}\right]$).

\item[iv)] If a message pair $(w_1,w_2)$ is to be transmitted, pick
one pair $(\bu^n(k), \bv^n(l)) \in A_\ep^{(n)}(U,V) \cap B_{w_1}
\times C_{w_2}$. Then, find an $\bx(w_1,w_2)$ which is jointly
$\ep$-typical with $(w_1,w_2)$ pair and designate it as the
corresponding codeword of $(w_1,w_2)$. Send over the BCC part,
$p(y_1,y_2|x)$.
\end{itemize}

\textbf{Decoding at BCC part:}
\begin{itemize}
\item[i)] Find the indexes $\hat{k}$ (resp. $\hat{l}$) such that $(\bu^n(\hat{k}),\by_1) \in
A_\ep^{(n)}(U,Y_1)$ (resp. $(\bv^n(\hat{l}),\by_2) \in
A_\ep^{(n)}(V,Y_2)$). If $\hat{k}, \hat{l}$ are not unique or does
not exist, declare an error, i.e. $\hW_1 = 0$ and/or $\hW_2 = 0$.
Else, decide $\hW_1 \in \cW_1$ (resp. $\hW_2 \in \cW_2$) at mobile
unit one (resp two), such that $\hat{k} \in B_{\hW_1}$ (resp.
$\hat{l} \in C_{\hW_2}$).
\end{itemize}

\textbf{Encoding at MAC part:}
\begin{itemize}
\item[i)]\underline{Generation of the codebook}(Nested codebook structure):
Fix $p(q_1), p(q_2)$. Let $p(q_1,q_2)=p(q_1)p(q_2)$. Generate the
$w_i$-th codebook $\mathcal{C}_{MAC}^{w_i} \in
\cQ_i^{2^{nR_i^{Data}}} \times \cQ_i^n$ such that $(j,k)$-th element
is $q_{w_i,k}(j)$ and $q_{w_i,k}(j)$s are i.i.d. realizations of
$Q_i$ of which distribution is $p(q_i)$ for all $j \in \{ 1, \ldots,
2^{nR_i^{Data}}\}$, $k \in \{ 1, \ldots, n\}$ and $i \in \{1,2\}$.

\item[ii)] Choose a message $M_i \in \cM_i^{\hW_i}$ uniformly for the $\hW_i$ decided at the BCC part,
i.e. $\Pr(M_i=m_i) = \frac{1}{2^{nR_i^{Data}}}$, for all $m_i \in
\cM_i^{\hW_i}$ and for $i \in \{1,2\}$. In order to send the message
$m_i$, pick the corresponding codeword $\bq_{\hW_i}^{n}(m_i)$ of
$\mathcal{C}_{MAC}^{\hW_i}$ and send over the imperfection channel
$p(\hq_i|q_{\hW_i})$ resulting in $\hbq_i^n$ for $i \in \{1,2\}$.
The pair of $(\hbq_1,\hbq_2)$ is the input to the MAC part,
$p(s|\hq_1,\hq_2)$.
\end{itemize}

\textbf{Decoding at MAC part:}
\begin{itemize}
\item[i)] Find the pair of indexes $\left( \hM_1,\hM_2 \right) \in \cM_1^{w_1} \times
\cM_2^{w_2}$ such that $(\bq_{w_1}^n(\hM_1), \bq_{w_2}^n(\hM_2),
\bs^n) \in A_\ep^{(n)}(Q_1,Q_2,S)$, where $A_\ep^{(n)}(Q_1,Q_2,S)$
is the $\ep$-typical set with respect to distribution
\begin{eqnarray}
p(q_1,q_2,s) & = & \sum_{\hq_1,\hq_2} p(s|\hq_1,\hq_2,q_1,q_2) p(\hq_1,\hq_2|q_1,q_2) p(q_1)p(q_2), \label{eq:average-prob-err-MAC-3} \\
 & = & \sum_{\hq_1,\hq_2} p(s|\hq_1,\hq_2) p(\hq_1,\hq_2|q_1,q_2) p(q_1)p(q_2), \label{eq:average-prob-err-MAC-4}\\
 & = & \sum_{\hq_1,\hq_2} p(s|\hq_1,\hq_2) p(\hq_1|q_1) p(\hq_2|q_2) p(q_1) p(q_2), \label{eq:average-prob-err-MAC-5}
\end{eqnarray}
where \eqref{eq:average-prob-err-MAC-3} follows since
$p(q_1,q_2)=p(q_1)p(q_2)$ (cf. the codebook generation of MAC part),
\eqref{eq:average-prob-err-MAC-4} follows since MAC channel depends
on only $(\hq_1,\hq_2)$ and \eqref{eq:average-prob-err-MAC-5}
follows since imperfection channels are independent and depends on
only $q_1$ and $q_2$, respectively.

If such a $\left( \hM_1, \hM_2 \right)$ pair does not exist or is
not unique, then declare an error, i.e. $\hM_1 = 0$ and/or $\hM_2 =
0$; otherwise decide $\left( \hM_1,\hM_2 \right)$.
\end{itemize}

\textbf{Analysis of Probability of Error:}\\
We begin with BCC part. By defining the error event as
$\mathcal{E}^{BCC} \eqdef \left\{ (\hW_1(\bY_1^n), \hW_2(\bY_2^n))
\neq (W_1,W_2)\right\}$, we have the following expression for the
average probability of error averaged over all messages,
$(w_1,w_2)$, and codebooks, $\cC_{BCC}$
\begin{eqnarray}
P_{e,BCC}^{(n)} & = & \Pr \left( \mathcal{E}^{BCC} \right), \nonumber \\
 & = & \Pr\left( \mathcal{E}^{BCC} | (W_1,W_2) = (1,1) \right),
 \label{eq:average-prob-err-BCC-1}
\end{eqnarray}
where \eqref{eq:average-prob-err-BCC-1} follows by noting the
equality of arithmetic average probability of error and the average
probability of error given in \eqref{eq:average-prob-err-BCC} and
the symmetry of the codebook construction at the BCC part.

Next, we define following type of error events:
\begin{eqnarray}
\mathcal{E}_1^{BCC} & \eqdef & \left\{\nexists (\bu^n(k), \bv^n(l)) \in (B_{1} \times C_{1}) \cap A_\ep^{(n)}(U,V)\right\}, \label{eq:average-prob-err-BCC-E1} \\
\mathcal{E}_2^{BCC} & \eqdef & \left\{ (\bu^n(k), \bv^n(l), \bx^n(w_1,w_2), \by_1^n, \by_2^n) \not \in A_\ep^{(n)}(U,V,X,Y_1,Y_2) \right\}, \label{eq:average-prob-err-BCC-E2}\\
\mathcal{E}_3^{BCC} & \eqdef & \left\{ \exists \hat{k} \neq k, \textrm{ s.t. } (\bu^n(\hat{k}),\by_1^n) \in A_\ep^{(n)}(U,Y_1)\right\}, \label{eq:average-prob-err-BCC-E3} \\
\mathcal{E}_4^{BCC} & \eqdef & \left\{ \exists \hat{l} \neq l,
\textrm{ s.t. } (\bv^n(\hat{l}),\by_2^n) \in A_\ep^{(n)}(V,Y_2)\right\},
\label{eq:average-prob-err-BCC-E4}
\end{eqnarray}
where \eqref{eq:average-prob-err-BCC-E1} corresponds to the failure
of the encoding, \eqref{eq:average-prob-err-BCC-E3} (resp.
\eqref{eq:average-prob-err-BCC-E4}) corresponds to the failure of
the decoding at mobile unit one (resp. mobile unit two).

Using typicality arguments, it can be shown that $\Pr \left(
\mathcal{E}_i^{BCC}\right) \leq \epsilon/4$ for $i \in \{ 2, 3, 4\}$
and Lemma 1 of \cite{gam-meu:81} also guarantees that $\Pr \left(
\mathcal{E}_1^{BCC} \right) \leq \epsilon/4$. Using these facts and
the union bound, we conclude that
\begin{equation}
P_{e,BCC}^{(n)} = \Pr(\mathcal{E}^{BCC}) =
\Pr(\mathcal{E}^{BCC}|(W_1,W_2)=(1,1)) \leq \ep,
\label{eq:average-prob-err-BCC-5}
\end{equation}
for any $\ep > 0$, for sufficiently large $n$; provided that
$I(U;Y_1) > R_1^{ID} + \ep$, $I(V;Y_2) > R_2^{ID} + \ep$, $ I(U;Y_1)
+ I(V;Y_2) - I(U;V) > R_1^{ID} + R_2^{ID} + 2 \ep + \delta(\ep)$,
such that $\delta(\ep) \rightarrow 0$ as $\ep \rightarrow 0$.

Further, using standard arguments for finding a code with negligible
maximal probability of error (cf. \cite{cov:06} pp. 203-204) from
the one with $P_{e,BCC}^{(n)} \leq \ep$ we conclude that we have
\begin{equation}
\lambda_{BCC}^{(n)} \eqdef \max_{w_1,w_2} \lambda_{BCC}^{w_1,w_2}
\leq 2 \ep, \label{eq:max-prob-err-BCC}
\end{equation}
for any $\ep >0$ and for sufficiently large $n$, which concludes the
BCC part.

By defining the error event as $\mathcal{E}^{MAC} \eqdef \left\{
\left( \hM_1(\bS^n),\hM_2(\bS^n) \right) \neq (M_1,M_2) | \left(
\hW_1, \hW_2 \right) = \left( w_1, w_2 \right) \right\}$, we have
the following expression for the average probability of error
averaged over all messages, $(m_1,m_2)$, and codebooks corresponding
to the messages, $\cC_{MAC}^{w_1}$ and $\cC_{MAC}^{w_2}$
\begin{eqnarray}
P_{e,MAC}^{(n)} & = & \Pr \left( \mathcal{E}^{MAC} \right), \nonumber \\
 & = & \Pr\left( \mathcal{E}^{MAC} | (M_1,M_2) = (1,1) \right),
 \label{eq:average-prob-err-MAC-1}
\end{eqnarray}
where \eqref{eq:average-prob-err-MAC-1} follows by noting the
equality of arithmetic average probability of error and the average
probability of error given in \eqref{eq:average-prob-err-MAC} and
the symmetry of the nested codebook construction at the MAC part.

Next, we define the following events
\begin{equation}
\mathcal{E}_{ij}^{MAC} \eqdef \left\{ (\bq_{w_1}^n (i), \bq_{w_2}^n
(j), \bs^n) \in A_\ep^{(n)}(Q_1,Q_2,S) \right\},
\label{eq:average-prob-err-MAC-E}
\end{equation}

Using union bound and appropriately bounding each error event by
exploiting typicality arguments, one can show that
\begin{equation}
P_{e,MAC}^{(n)} = \Pr\left(\mathcal{E}^{MAC}\right) =
\Pr\left(\mathcal{E}^{MAC} | (M_1,M_2) = (1,1) \right)\leq \ep,
\label{eq:average-prob-err-MAC-6}
\end{equation}
for any $\ep>0$ and sufficiently large $n$; provided that
$I(Q_1;S|Q_2) - R_1^{Data}
> 3 \ep$, $I(Q_2;S|Q_1) - R_2^{Data} > 3 \ep$ and $I(Q_1,Q_2;S) -
(R_1^{Data} + R_2^{Data})> 4 \ep$.

Further, using standard arguments for finding a code with negligible
maximal probability of error (cf. \cite{cov:06} pp. 203-204) from
the one with $P_{e,MAC}^{(n)} \leq \ep$ we conclude that we have
\begin{equation}
\lambda_{MAC}^{(n)} \eqdef \max_{m_1,m_2} \lambda_{MAC}^{m_1,m_2}
\leq 2 \ep, \label{eq:max-prob-err-MAC}
\end{equation}
for any $\ep >0$ and for sufficiently large $n$, which concludes the
MAC part.

Next, we sum up things and conclude the proof in the following
manner.

First, by plugging \eqref{eq:cond-prob-err-overall} in
\eqref{eq:problem-statement-5-max-prob-err}, we have
\begin{equation}
\lambda^{(n)} =
\max_{\lambda_{BCC}^{w_1,w_2},\lambda_{MAC}^{m_1,m_2}}
\lambda_{BCC}^{w_1,w_2} + \lambda_{MAC}^{m_1,m_2} -
\lambda_{BCC}^{w_1,w_2}\lambda_{MAC}^{m_1,m_2}.
\label{eq:max-prob-err-overall-1}
\end{equation}
Further, using the fact that the cost function in
\eqref{eq:max-prob-err-overall-1} is monotonic increasing in both
$\lambda_{BCC}^{w_1,w_2}$ and $\lambda_{MAC}^{m_1,m_2}$, we conclude
that (cf. \eqref{eq:max-prob-err-BCC} and
\eqref{eq:max-prob-err-MAC})
\begin{equation}
\lambda^{(n)} \leq 4 \ep - 4 \ep^2,
\label{eq:max-prob-err-overall-2}
\end{equation}
for any $0 < \ep < 1$ and sufficiently large $n$. Since $\ep$ may be
arbitrarily small, \eqref{eq:max-prob-err-overall-2} concludes the
proof.
\end{proof}

\section{Power Constrained Gaussian Case}
\label{sec:gaussian-case}
\subsection{Problem Statement}
\label{ssec:problem-statement-gaussian}

In this section, we generalize the communication problem stated in
Section~\ref{ssec:problem-statement} to continuous random variables
under the assumption of Gaussian noise and power constraint on the
codebooks. To be more precise we have the problem depicted in
Figure~\ref{fig:gaussian}, with the power constraints:
\begin{eqnarray}
\textrm{E}\left[ X^2 \right] & \leq & P, \label{eq:gaussian-power-const-1}\\
\textrm{E}\left[ (Q_{1,\hW_1})^2 \right] & \leq & \alpha_1 P_1, \label{eq:gaussian-power-const-2}\\
\textrm{E}\left[ (Q_{1,\hW_2})^2 \right] & \leq & \alpha_2 P_2,
\label{eq:gaussian-power-const-3}
\end{eqnarray}
such that $\alpha_1,\alpha_2 < 1$ and $P_1+P_2 \leq P$, where $P_1$
(resp. $P_2$) is the power delivered to mobile unit one (resp. two)
and w.l.o.g. we assume that $N_1 < N_2$.
\begin{figure}[!htb]
\centering \epsfxsize 6in
 \epsfbox{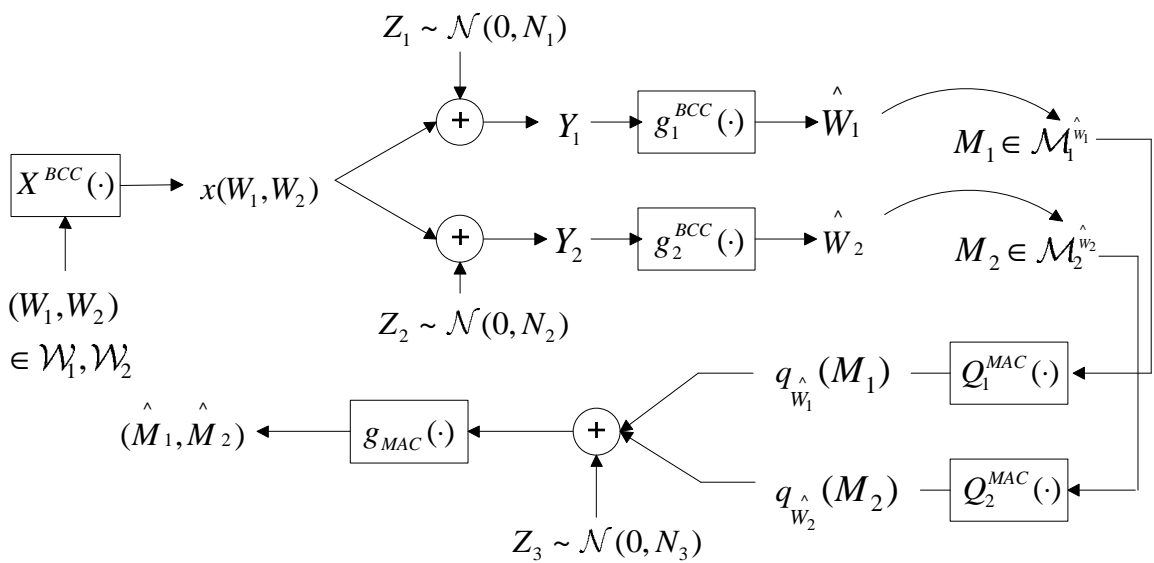} \caption{Block Diagram Representation of the multiuser communication system under Gaussian noise assumption.} \label{fig:gaussian}
\end{figure}

Note that both Definition~\ref{def:communication-system} (excluding
imperfection channels, which are irrelevant for this case) and
Definition~\ref{def:code} are valid for this case, with $\cX = \cQ_1
= \cQ_2 = \cS = \bbR$.

\begin{remark}
$ $
\begin{itemize}
\item[(i)] Observe that, we model the ``imperfection channel'' of
discrete case as an additional power constraint for the Gaussian
case.

\item[(ii)] BCC part for the Gaussian case at hand is equivalent to ``degraded BCC'',
which enables us to state the \emph{capacity region} instead of
characterizing achievable region only.
\end{itemize}
\label{rem:gaussian}
\end{remark}

\subsection{Capacity Region for Gaussian Case}
\label{ssec:gaussian-capacity-region}

In this section, we state the capacity region of the communication
system given in Section~\ref{ssec:problem-statement-gaussian}. Note
that throughout the section, all the logarithms are base $e$, in
other words the unit of information is ``nats''.

\begin{theorem}
The capacity region, $\cR_1 \subset \bbR^4$, of the system shown in
Figure~\ref{fig:gaussian} is given by
\begin{eqnarray}
\cR_1 \eqdef \left\{ (R_1^{ID}, R_2^{ID}, R_1^{Data}, R_2^{Data}) \;
: \; R_1^{ID}, R_2^{ID}, R_1^{Data}, R_2^{Data} \geq 0, \; R_1^{ID} < \frac{1}{2}\log\left(1+\frac{\alpha P}{N_1}\right),\right. \nonumber\\
R_2^{ID} < \frac{1}{2} \log\left( 1 + \frac{(1-\alpha)P}{N_2+\alpha P}\right), \; R_1^{Data} < \frac{1}{2} \log\left( 1 + \frac{\alpha \alpha_1 P}{N_3}\right), \; R_2^{Data} < \frac{1}{2}\log\left( 1+\frac{(1-\alpha)\alpha_2 P}{N_3}\right),\nonumber\\
\left. R_1^{Data}+R_2^{Data} < \frac{1}{2} \log\left( 1+\frac{\alpha
\alpha_1 P + (1-\alpha)\alpha_2 P} {N_3}\right), \textrm{ s. t. } 0
\leq \alpha \leq 1, \; 0 \leq \alpha_1,\alpha_2 \leq 1 \right\},
\label{eq:gaussian-capacity-region}
\end{eqnarray}
where $\alpha$ may be chosen arbitrarily in the given range and
$\alpha_1$ and $\alpha_2$ are system parameters.
\label{thrm:capacity-region-gaussian}
\end{theorem}

\subsubsection{Achievability} \label{ssec:gaussian-achievability}

In section, we prove the forward part of
Theorem~\ref{thrm:capacity-region-gaussian}, in other words
following theorem:

\begin{theorem}
Any rate quadruple $(R_1^{ID}, R_2^{ID}, R_1^{Data}, R_2^{Data}) \in
\bbR^4$, there exists a sequence of \\
$\left(2^{nR_1^{ID}},
2^{nR_2^{ID}}, 2^{nR_1^{Data}}, 2^{nR_2^{Data}}, n\right)$ codes
with arbitrarily small probability of error for sufficiently large
$n$, provided that
\begin{eqnarray}
\frac{1}{2} \log \left( 1 + \frac{\alpha P}{N_1}\right) & > &  R_1^{ID} + \ep, \label{eq:gaussian-achievability-1}\\
\frac{1}{2} \log \left( 1 + \frac{(1 - \alpha)P}{\alpha P + N_2}\right) & > &  R_2^{ID} + \ep, \label{eq:gaussian-achievability-2}\\
\frac{1}{2} \log \left( 1 + \frac{\alpha_1 \alpha P }{N_3} \right) & > &  R_1^{Data} + 3 \ep , \label{eq:gaussian-achievability-3}\\
\frac{1}{2} \log \left( 1 + \frac{\alpha_2 (1- \alpha) P }{N_3}\right) & > &  R_2^{Data} + 3 \ep, \label{eq:gaussian-achievability-4}\\
\frac{1}{2} \log \left( 1 + \frac{\alpha_1 \alpha P + \alpha_2 (1-
\alpha) P}{N_3}\right)& > &  R_1^{Data} + R_2^{Data} + 4 \ep ,
\label{eq:gaussian-achievability-5}
\end{eqnarray}
for any $\ep >0$, $0 \leq \alpha \leq 1$ and $0 \leq \alpha_1,
\alpha_2 \leq 1$. \label{thrm:gaussian-achievability}
\end{theorem}

\begin{proof}
In order to prove the theorem, we use \emph{superposition coding}
\cite{cov:06} at BCC part and standard random coding at MAC part. W.l.o.g. suppose $\epsilon \in (0, 13/84)$. \footnote{Since we want to show that $\lambda^{(n)} \rightarrow 0$ as $n \rightarrow \infty$, this will suffice. To see this, observe that in the proof of the theorem, we show that for any sufficiently large $n$ and for any $\epsilon \in (0,13/84)$, $\lambda^{(n)} \leq \epsilon$, which directly implies  $\lambda^{(n)} \leq \epsilon^\prime$ for any $\epsilon^{\prime} \geq 13/84$.}

\textbf{Encoding at BCC part:}
\begin{itemize}
\item[i)]\underline{Generation of the codebook:}(Superposition Coding) Generate codebook, $\cC_{BCC}^1$ (resp. $\cC_{BCC}^2$) with corresponding rate
$R_1^{ID}$ (resp. $R_2^{ID}$) such that both $R_1^{ID}$ and
$R_2^{ID}$ satisfy the conditions
\eqref{eq:gaussian-achievability-1},
\eqref{eq:gaussian-achievability-2} and
\eqref{eq:gaussian-achievability-3} where
\begin{equation}
\cC_{BCC}^1 \eqdef \left[ x_{1,i}(w_1)\right],
\label{eq:gaussian-BCC-codebook-1}
\end{equation}
such that each $x_{1,i}(w_1)$ are i.i.d. realizations of $X_1 \sim
\cN (0, \alpha P - \ep/2)$ and
\begin{equation}
\cC_{BCC}^2 \eqdef \left[ x_{2,i}(w_2)\right],
\label{eq:gaussian-BCC-codebook-2}
\end{equation}
such that each $x_{2,i}(w_2)$ are i.i.d. realizations of $X_2 \sim
\cN (0, (1-\alpha)P -\ep/2)$. Reveal both $\cC_{BCC}^1$ and
$\cC_{BCC}^2$ to each mobile unit.

\item[ii)] Choose a message pair $(w_1,w_2) \in \cW_1 \times \cW_2$,
uniformly over $\cW_1 \times \cW_2$, i.e. $\Pr(W_1=w_1,W_2=w_2) =
1/2^{n(R_1^{ID}+R_2^{ID})}$, for all $(w_1,w_2) \in \cW_1 \times
\cW_2$.

\item[iii)] In order to send message $(w_1,w_2)$, take
$\bx_1^n(w_1)$ from $\cC_{BCC}^1$ and $\bx_2^n(w_2)$ from
$\cC_{BCC}^2$ and send $\bx^n(w_1,w_2) \eqdef \bx_1^n (w_1) +
\bx_2^n(w_2)$ over the BCC to both sides, yielding $Y_1 \eqdef
\bx^n(w_1,w_2) + Z_1$ at mobile unit one and $Y_2 \eqdef
\bx^n(w_1,w_2) + Z_2$ at mobile unit two, where $Z_1$ and $Z_2$ are
arbitrarily correlated with following marginal distributions $Z_1
\sim \cN(0,N_1)$, $Z_2 \sim \cN(0,N_2)$. Note that law of large
numbers ensures $\bx^n(w_1,w_2)$ satisfies the power constraint of
\eqref{eq:gaussian-power-const-1}.
\end{itemize}

\textbf{Decoding at BCC part:}
\begin{itemize}
\item[i)] Upon receiving $\by_2^n$, second mobile unit performs
jointly typical decoding, i.e. decides the unique $\hW_2 \in \cW_2$
such that $\left( \by_2^n, \bx_2^n(\hW_2)\right) \in
A_\ep^{(n)}(X_2,Y_2)$. If such a $\hW_2 \in \cW_2$ does not exist or
is not unique, then declares an error, i.e. $\cW_2 = 0$.

Mobile unit one also performs the same jointly typical decoding
first with $\by_1^n$ in order to decide the unique $\hW_2 \in \cW_2$
such that $\left( \by_1^n, \bx_1^n(\hW_2)\right) \in
A_\ep^{(n)}(X_2,Y)$. If such $\hW_2 \in \cW_2$ does not exist or is
not unique, then declares an error, i.e. $\cW_2 = 0$. After deciding
on $\hW_2$, mobile unit one calculates the corresponding $\by^n
\eqdef \by_1^n - \bx_2^n(\hW_2)$ and then performs jointly typical
decoding, i.e. decides the unique $\hW_1 \in \cW_1$ such that
$\left( \by^n, \bx_1^n(\hW_1)\right) \in A_\ep^{(n)}(X_1,Y)$. If
such a $\hW_1 \in \cW_1$ does not exist or is not unique, then
declares an error, i.e. $\hW_1 = 0$.
\end{itemize}

\textbf{Encoding at MAC part:}
\begin{itemize}
\item[i)]\underline{Generation of Codebook} (Nested Codebook Structure):
Fix $f(q_1), f(q_2)$. Let $f(q_1,q_2)=f(q_1)f(q_2)$. Generate the
$w_1$-th (resp. $w_2$-th) codebook as $\mathcal{C}_{MAC}^{w_1}
\eqdef \left[q_{w_1,j}(m_1) \right]$ (resp. $\mathcal{C}_{MAC}^{w_2}
\eqdef \left[q_{w_2,j}(m_2) \right]$), such that $q_{w_1,j}(m_1)$
(resp. $q_{w_2,j}(m_2)$) are i.i.d. realizations of $Q_1 \sim
\cN(0,\alpha_1 \alpha P - \ep)$ (resp. $Q_2 \sim \cN(0, \alpha_2
(1-\alpha) P - \ep)$) for all $w_1 \in \{ 1, \ldots,
2^{nR_1^{ID}}\}$ (resp. $w_2 \in \{ 1, \ldots, 2^{nR_2^{ID}}\}$),
$m_1 \in \{ 1, \ldots, 2^{nR_1^{Data}}\}$ (resp. $m_2 \in \{ 1,
\ldots, 2^{nR_2^{Data}}\}$) and $j \in \{ 1, \ldots, n\}$.

\item[ii)] Choose a message $M_i \in \cM_i^{\hW_i}$ uniformly, i.e.
$\Pr(M_i = m_i) = 1/2^{nR_i^{Data}}$, for all $m_i \in
\cM_i^{\hW_i}$ and for $i \in \{1,2\}$. In order to send a message
$m_i$, take the corresponding codeword $\bq_{\hW_i}^n$ of
$\cC_{MAC}^{\hW_i}$ and send over the MAC, for $i \in \{1,2\}$,
resulting in $\bS^n \eqdef \bq_{\hW_1}^n + \bq_{\hW_2}^n + \bZ_3^n$.
\end{itemize}

\textbf{Decoding at MAC part:}
\begin{itemize}
\item[i)] Find the pair of indexes $(\hM_1,\hM_2) \in \cM_1^{w_1} \times
\cM_2^{w_2}$ such that $(\bq_{w_1}(\hM_1), \bq_{w_2}(\hM_2), \bs^n)
\in A_\ep^{(n)}(Q_1,Q_2,S)$. If such a pair does not exist or is not
unique, then declare an error, i.e. $\hM_1 = 0$ and/or $\hM_2 = 0$;
otherwise decide $(\hM_1,\hM_2)$.
\end{itemize}

\textbf{Analysis of Probability of Error:} We begin with the BCC
part. First, note that \eqref{eq:average-prob-err-BCC-1} is still
valid as well as the error event definition. Next, we define
following type of error events
\begin{eqnarray}
\mathcal{E}_0^{BCC} & \eqdef & \left\{ \frac{1}{n}\sum_{j=1}^n x_j^2(1,1) > P\right\}, \label{eq:gaussian-BCC-err-event-0} \\
\mathcal{E}_{1,i}^{BCC} & \eqdef & \left\{ (\bx_2^n(i),\by_1^n) \in A_\ep^{(n)}(X_2,Y_1), \textrm { s.t. } i \neq 1 \right\}, \label{eq:gaussian-BCC-err-event-1}\\
\mathcal{E}_{2,j}^{BCC} & \eqdef & \left\{ (\bx_1^n(j),\by^n) \in A_\ep^{(n)}(X_1,Y), \textrm { s.t. } j \neq 1 \right\}, \label{eq:gaussian-BCC-err-event-2}\\
\mathcal{E}_{3,k}^{BCC} & \eqdef & \left\{ (\bx_2^n(k),\by_2^n) \in
A_\ep^{(n)}(X_2,Y_2), \textrm { s.t. } k \neq 1 \right\},
\label{eq:gaussian-BCC-err-event-3}
\end{eqnarray}
where \eqref{eq:gaussian-BCC-err-event-0} corresponds to the
violation of the power constraint,
\eqref{eq:gaussian-BCC-err-event-1} corresponds to the failure of
the first step of the decoding at the mobile unit one,
\eqref{eq:gaussian-BCC-err-event-2} corresponds to the failure of
the second step of the decoding at the mobile unit one,
\eqref{eq:gaussian-BCC-err-event-3} corresponds to the failure of
the decoding at the mobile unit two.

Using union bound and appropriately bounding the probability of each
error event term by using arguments of typicality (except for the
power constraint, which follows from law of large numbers), one can
show that
\begin{equation}
P_{e,BCC}^{(n)} = \Pr \left( \mathcal{E}^{BCC} \right) = \Pr \left(
\mathcal{E}^{BCC} | (W_1,W_2) = (1,1) \right) \leq 7 \ep,
\label{eq:gaussian-BCC-average-prob-err-31}
\end{equation}
for any $\ep >0$ and sufficiently large $n$, provided that
$\frac{1}{2}\log\left( 1 + \frac{\alpha P}{N_1}\right) - R_1^{ID}
> \ep$ (cf. \eqref{eq:gaussian-achievability-1}), $\frac{1}{2}\log\left( 1 + \frac{(1 - \alpha)P}{\alpha P +
N_2}\right) - R_2^{ID} > \ep$ (cf.
\eqref{eq:gaussian-achievability-2}) and $\frac{1}{2} \log \left( 1
+ \frac{(1 - \alpha)P}{\alpha + N_1} \right) - R_2^{ID} > \ep$ (
which is guaranteed by recalling $N_1<N_2$ and
\eqref{eq:gaussian-achievability-1}.

Further, using standard arguments for finding a code with negligible
maximal probability of error (cf. \cite{cov:06} pp. 203-204) from
the one with $P_{e,BCC}^{(n)} \leq 7 \ep$ we conclude that we have
\begin{equation}
\lambda_{BCC}^{(n)} \eqdef \max_{w_1,w_2} \lambda_{BCC}^{w_1,w_2}
\leq 14 \ep, \label{eq:gaussian-max-prob-err-BCC}
\end{equation}
for any $\ep >0$ and sufficiently large $n$, provided that
\eqref{eq:gaussian-achievability-1} and
\eqref{eq:gaussian-achievability-2} hold, which concludes the BCC
part.

Now, we continue with the MAC part and note that
\eqref{eq:average-prob-err-MAC-1} is still valid as well as the
error event definition. We additionally include the following type
of error event, which deals with the power constraints
\begin{equation}
\mathcal{E}^{MAC}_{0,i} \eqdef \left\{ \frac{1}{n}\sum_{j=1}^n
q_{w_i,j}^2(1) > \alpha_i P_i \right \},
\label{eq:gaussian-MAC-err-event-0}
\end{equation}
for $i \in \{ 1, 2\}$, such that $P_1 = \alpha P$ and $P_2 = (1 -
\alpha)P$ and $\alpha$ is the same as the one given in BCC case.

Using union bound and appropriately bounding the probability of each
error event term by using arguments of typicality (except for the
power constraint related terms, which follow from law of large
numbers), one can show that
\begin{equation}
P_{e,MAC}^{(n)} = \Pr \left( \mathcal{E}^{MAC} \right) = \Pr \left(
\mathcal{E}^{MAC} | (M_1,M_2) = (1,1) \right) \leq 6 \ep,
\label{eq:gaussian-BCC-average-prob-err-3}
\end{equation}
for any $\ep >0$ and sufficiently large $n$, provided that
$\frac{1}{2} \log \left( 1 + \frac{\alpha_1 \alpha P }{N_3} \right)
>  R_1^{Data} + 3 \ep $,
$\frac{1}{2} \log \left( 1 + \frac{\alpha_2 (1- \alpha) P
}{N_3}\right) >  R_2^{Data} + 3 \ep$, $\frac{1}{2} \log \left( 1 +
\frac{\alpha_1 \alpha P + \alpha_2 (1- \alpha) P}{N_3}\right) >
R_1^{Data} + R_2^{Data} + 4 \ep$.

Further, using standard arguments for finding a code with negligible
maximal probability of error (cf. \cite{cov:06} pp. 203-204) from
the one with $P_{e,MAC}^{(n)} \leq 6 \ep$ we conclude that we have
\begin{equation}
\lambda_{MAC}^{(n)} \eqdef \max_{m_1,m_2} \lambda_{MAC}^{m_1,m_2}
\leq 12 \ep, \label{eq:gaussian-max-prob-err-MAC}
\end{equation}
for any $\ep >0$ and sufficiently large $n$, provided that
\eqref{eq:gaussian-achievability-3},
\eqref{eq:gaussian-achievability-4} and
\eqref{eq:gaussian-achievability-5} hold, which concludes the MAC
part.

Following similar arguments as in
Section~\ref{ssec:achievable-region-discrete-case} and using
\eqref{eq:gaussian-max-prob-err-BCC} and
\eqref{eq:gaussian-max-prob-err-MAC}, we conclude that
\begin{equation}
\lambda^{(n)} \leq \ep (26 - 168 \ep),
\label{eq:gaussian-max-prob-err-overall}
\end{equation}
for any $ 0 < \ep <\frac{13}{84}$, where $\lambda^{(n)}$ is as
defined in \eqref{eq:problem-statement-5-max-prob-err}. Since $\ep$
may be arbitrarily small, \eqref{eq:gaussian-max-prob-err-overall}
concludes the proof.
\end{proof}

\subsubsection{Converse}

In this section, we prove the converse part of
Theorem~\ref{thrm:capacity-region-gaussian}, in other words we have
the following theorem:
\begin{theorem}
For any sequence of $\left(2^{nR_1^{ID}}, 2^{nR_2^{ID}},
2^{nR_1^{Data}}, 2^{nR_2^{Data}},n\right)$-RFID codes with
$P_e^{(n)} < \ep$, for any $\ep >0$, we have $\left( R_1^{ID},
R_2^{ID}, R_1^{Data}, R_2^{Data}\right) \in \cR_1$.
\label{thrm:gaussian-converse}
\end{theorem}

\begin{proof}
Proof relies on ideas from \cite{ber:74} for BCC part and
\cite{cov:06} for MAC part.

First of all, we have following
\begin{eqnarray}
P_e^{(n)} & = & 1- \Pr\left( \left[(\hW_1,\hW_2) = (W_1,W_2)\right] \wedge \left[(\hM_1,hM_2) = (M_1,M_2)\right] \right), \nonumber\\
 & = & 1- \Pr\left( (\hW_1,\hW_2) = (W_1,W_2) \right) \Pr\left((\hM_1,\hM_2) = (M_1,M_2)|(\hW_1,\hW_2) = (W_1,W_2)
 \right).\label{eq:gaussian-converse-prob-err-1}
\end{eqnarray}
Using \eqref{eq:gaussian-converse-prob-err-1} and noting that
$P_e^{(n)} \leq \ep$, we have
\begin{equation*}
\left( 1- \Pr\left( (\hW_1,\hW_2) \neq (W_1,W_2)
\right)\right)\left( \Pr\left((\hM_1,\hM_2) \neq
(M_1,M_2)|(\hW_1,\hW_2) = (W_1,W_2)\right)\right),
\end{equation*}
which implies
\begin{equation}
P_{e,BCC}^{(n)} = \Pr\left( (\hW_1,\hW_2) \neq (W_1,W_2) \right)
\leq \ep, \label{eq:gaussian-converge-prob-err-BCC}
\end{equation}
and
\begin{equation}
P_{e,MAC}^{(n)} = \Pr\left( (\hM_1,\hM_2) \neq (M_1,M_2) |
(\hW_1,\hW_2) = (W_1,W_2) \right) \leq \ep,
\label{eq:gaussian-converge-prob-err-MAC}
\end{equation}
Next, \eqref{eq:gaussian-converge-prob-err-BCC} enables us to use
the result of \cite{ber:74} for BCC case, hence we state that
\begin{eqnarray}
R_{1}^{ID} & \leq & \frac{1}{2}\log\left( 1 + \frac{\alpha P}{N_1}\right),\label{eq:gaussian-converse-BCC-eq-1}\\
R_{2}^{ID} & \leq & \frac{1}{2}\log\left( 1 + \frac{(1 -
\alpha)P}{\alpha P +
N_2}\right),\label{eq:gaussian-converse-BCC-eq-2}
\end{eqnarray}
for any $0 \leq \alpha \leq 1$.

Further, \eqref{eq:gaussian-converge-prob-err-MAC} enables us to use
the result of \cite{cov:06} for MAC case, hence we state that
\begin{eqnarray}
R_1^{Data} & \leq & \frac{1}{2}\log\left( 1 + \frac{\alpha_1 \alpha P}{N_3}\right), \label{eq:gaussian-converse-MAC-eq-1}\\
R_2^{Data} & \leq & \frac{1}{2}\log\left( 1 + \frac{\alpha_2 (1 - \alpha) P}{N_3}\right), \label{eq:gaussian-converse-MAC-eq-2}\\
R_1^{Data} + R_2^{Data} & \leq & \frac{1}{2}\log\left( 1 +
\frac{\alpha_1 \alpha P + \alpha_2 (1 - \alpha) P}{N_3}\right).
\label{eq:gaussian-converse-MAC-eq-3}
\end{eqnarray}

Combining \eqref{eq:gaussian-converse-BCC-eq-1},
\eqref{eq:gaussian-converse-BCC-eq-2},
\eqref{eq:gaussian-converse-MAC-eq-1},
\eqref{eq:gaussian-converse-MAC-eq-2} and
\eqref{eq:gaussian-converse-MAC-eq-3} we conclude that for any
$\left(2^{nR_1^{ID}}, 2^{nR_2^{ID}}, 2^{nR_1^{Data}},
2^{nR_2^{Data}},n\right)$-RFID codes with $P_e^{n}$, we have $\left(
R_1^{ID}, R_2^{ID}, R_1^{Data}, R_2^{Data}\right) \in \cR_1$, which
concludes the proof.
\end{proof}

\section{Conclusion}
\label{sec:conclusion}
In this paper, we studied the RFID capacity problem by modeling the
underlying structure as a specific multiuser communication system that
is represented by a cascade of a BCC and a MAC. The BCC and MAC parts
are used to model communication between the RFID reader and the mobile
units, and between the mobile units and the RFID reader,
respectively. To connect the BCC and MAC parts, we used a ``nested
codebook'' structure. We further introduced imperfection channels for discrete alphabet case
as well as additional power limitations for continuous alphabet additive Gaussian noise 
case to accurately model the physical medium of the RFID system. We provided the achievable rate region in
the discrete alphabet case and the capacity region for the continuous
alphabet additive Gaussian noise case. Hence, overall, we characterized the maximal achievable
error free communication rates for any RFID protocol for the latter case.

\end{document}